\documentclass[oneside]{amsart}
\pdfoutput=1

\usepackage{geometry}
\usepackage[T1]{fontenc}
\usepackage[utf8]{inputenc}
\usepackage{lmodern}
\usepackage[dvipsnames]{xcolor}
\usepackage{tikz,tikz-3dplot}
\usetikzlibrary{positioning,arrows.meta,decorations.pathreplacing,decorations.pathmorphing,decorations.markings}

\newcommand{\ZZ}{{\mathbb Z}}

\newcommand{\ii}{\mathrm{i}}

\newcommand{\Tr}{\mathrm{Tr}\,}
\newcommand{\slN}{\mathrm{sl}(N)}

\theoremstyle{plain}
\newtheorem{thm}{Theorem}
\newtheorem{lem}[thm]{Lemma}
\newtheorem{con}[thm]{Conjecture}

\newtheorem{prop}[thm]{Proposition}

\newtheorem{defn}[thm]{Definition}
\newtheorem{ex}[thm]{Example}

\title{Notes on color reductions and $\gamma$ traces}
\date{}
\author{Oliver Schnetz}
\address{Oliver Schnetz\\
II. Institut f\"ur Theoretische Physik\\
Luruper Chaussee 149\\
22761 Hamburg, Germany}
\email{schnetz@mi.uni-erlangen.de}

\tikzset{
	ad/.style={line width=1pt},
	fun/.style={line width=1pt, postaction={decorate},
		decoration={markings,mark=at position .55 with {\arrow[scale=.5,>=triangle 45]{>}}}},
}

\begin{document}

\begin{abstract}
We present efficient algorithms to calculate the color factors for the $SU(N)$ gauge
group and to evaluate $\gamma$ traces. The aim of these notes is to give a self-contained,
proved account of the basic results with particular emphasis on color reductions. We fine tune existing algorithms to make calculations at high loop orders possible.
\end{abstract}
\maketitle

\section{Introduction}
Yang-Mills quantum field theories (QFTs) feature combinatorial factors in the
Feynman integrals which come from the $SU(N)$ color gauge group. The calculation of
these color factors is explained in \cite{Cvit} (see also \cite{Bondi,Haber}).

In QFTs with fermions, a first step in the calculation of a
Feynman integral often is the evaluation of the traces over the $\gamma$ matrices which
originate from fermion edges and from vertices. The evaluation of these
$\gamma$ traces is explained in \cite{Kennedy,Vasilev} (and in any textbook on QFT, see e.g.\ \cite{IZ}).

Today, one typically calculates to loop orders $\leq5$, where any implementation
of these reductions is sufficient, see e.g.\ \cite{FormColor,Form}.
With the method of graphical functions \cite{gf,gfe,numfunct,7loops,5lphi3} it may become
possible to tackle higher loop orders in certain setups.
At loop orders $\geq6$ it becomes increasingly desirable to fine tune
the algorithms for color and $\gamma$ reductions.

In these notes, we collect identities that are necessary to perform the reductions
to high loop orders. We include the proofs of all relevant results.
Particular emphasis is on color reductions, where the proof of the essential identity
(Proposition \ref{prop:color}) is not included in \cite{Cvit} (see \cite{Haber} for the proof).
We also prove some additional results about color reductions.

The suggested algorithms are implemented in the Maple package {\tt HyperlogProcedures}.
Typical color reductions are more or less instant
at relevant loop orders. The reduction of $\gamma$ traces is slightly more time-consuming, mostly because it can produce lengthy results at high loop orders.
For example, using {\tt HyperlogProcedures} on a single core of an office PC,
the average time for a $\gamma$ reduction of a Feynman graph that contributes to the photon propagator
is approximately 2 minutes at six loops and 30 minutes at seven loops.

\section*{Acknowlegements}
The author is supported by the DFG-grant SCHN 1240/3-1. He thanks Sven-Olaf Moch
for discussions and encouragement. The author also thanks Simon Theil for generating the figures in this work.

\section{Color reduction}
\subsection{The color graph}
We follow the algorithm presented in \cite{Cvit}, which is particularly simple in the case of the group SU($N$).
The complex Lie-algebra of SU($N$) is $\slN$, the Lie-algebra of traceless matrices.
For any representation of $\slN$, the commutator of the basis $T^i=(T^i)_{ab}$ is\footnote{In \cite{Cvit} $f_{ijk}=\ii C_{ijk}$. We do not see the necessity to
pass to complex numbers. Euclidean QED and Yang-Mills theory can be formulated in a real setup.}

\begin{equation}\label{color1}
[T^i,T^j]=f_{ijk}T^k=\sum_{k=1}^{N^2-1}f_{ijk}T^k,
\end{equation}
where we sum over repeated indices (Einstein's sum convention).
The structure constants $f_{ijk}=(f_i)_{kj}$ define an adjoint representation of $\slN$.

By definition, the trace of $T^i$ vanishes,
\begin{equation}\label{color2}
T^i_{aa}\equiv\sum_{a=1}^NT^i_{aa}=0.
\end{equation}
In some sense, we consider color factors as objects in a combinatorial differential geometry with Euclidean metric.

The matrices $T^i$ are chosen orthogonal and normalized,
\begin{equation}\label{color3}
\Tr T^iT^j=\delta_{i,j}.
\end{equation}
We use orthogonality to express the structure constants $f_{ijk}$ in terms of the matrices $T^i$,
\begin{equation}\label{color4}
f_{ijk}=\Tr [T^i,T^j]T^k.
\end{equation}
The cyclicity of the trace implies that the $f_{ijk}$ are fully anti-symmetric in their indices,
\begin{equation}\label{color5}
f_{ijk}=f_{jki}=f_{kij}=-f_{ikj}=-f_{kji}=-f_{jik}.
\end{equation}
The first two identities of (\ref{color5}) allow us to use a graphical representation.

\begin{figure}
\begin{tikzpicture}
\begin{scope}
	\coordinate (AA) at (-1,0);
	\coordinate (XX) at (0,0);
	\coordinate (BB) at (1,0);
	\coordinate (II) at (0,1);
	\draw[fill = black] (XX) circle (2 pt);
	\draw[fun](AA) -- (XX);
	\draw[fun](XX) -- (BB);
	\draw[ad](XX) -- (II);
	\node[above=.01 of AA,scale=1.1] {\textit{a}};
	\node[above=.01 of BB,scale=1.1] {\textit{b}};
	\node[right=.01 of II,scale=1.1] {\textit{i}};
\end{scope}

\begin{scope}[xshift=100]
	\coordinate (II) at (-1,0);
	\coordinate (JJ) at (1,0);
	\coordinate (KK) at (0,1.732);
	\coordinate (XX) at (0,0.577);
	\draw[fill = black] (XX) circle (1.5pt);
	\draw[ad](II) -- (XX);
	\draw[ad](JJ) -- (XX);
	\draw[ad](KK) -- (XX);
	\node[above=.01 of II,scale=1.1] {\textit{j}};
	\node[above=.01 of JJ,scale=1.1] {\textit{k}};
	\node[right=.01 of KK,scale=1.1] {\textit{i}};
\end{scope}
	
\begin{scope}[xshift=200]
	\coordinate (II) at (-1,0);
	\coordinate (JJ) at (1,0);
	\draw[ad](II) -- (JJ);
	\node[above=.01 of II,scale=1.1] {\textit{i}};
	\node[above=.01 of JJ,scale=1.1] {\textit{j}};
\end{scope}
	
\begin{scope}[xshift=300]
	\coordinate (AA) at (-1,0);
	\coordinate (BB) at (1,0);
	\draw[fun](AA) -- (BB);
	\node[above=.01 of AA,scale=1.1] {\textit{a}};
	\node[above=.01 of BB,scale=1.1] {\textit{b}};
\end{scope}
	
	\begin{scope}[yshift=-18]
	\begin{scope}
		\node[scale=1.1]{\large $T^i_{ab}$};
	\end{scope}
	\begin{scope}[xshift=100]
		\node[scale=1.1]{\large $f_{ijk}$};
	\end{scope}
	\begin{scope}[xshift=200]
		\node[scale=1.1]{\large $\delta_{i,j}$};
	\end{scope}
	\begin{scope}[xshift=300]
		\node[scale=1.1]{\large $\delta_{a,b}$};
	\end{scope}
	\end{scope}
	
\end{tikzpicture}
\caption{Feynman rules for color graphs}
\label{figcolor1}
\end{figure}
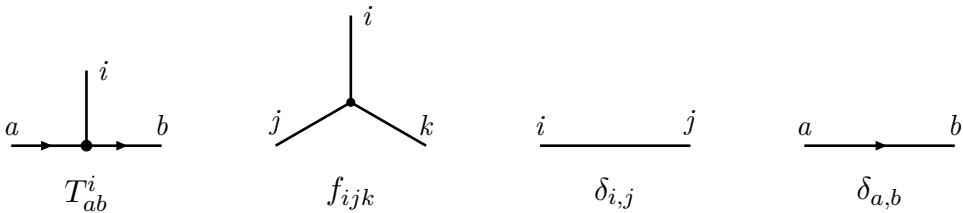

To $T^i_{ab}$ we associate a corolla (a vertex with three half-edges). The half-edge $i$ has no orientation, while the
half-edges $a$ and $b$ are ingoing and outgoing, respectively (see Figure \ref{figcolor1}). 

To $f_{ijk}$ we associate
a corolla of three non-oriented half-edges, where we fix a planar representation with counter-clockwise $i,j,k$.
The sum over double indices glues the corresponding half-edges. Note that matrix products in the $T^i$ preserve the orientation.
We obtain a graph with fixed planar embedding: a ribbon graph.
By anti-symmetry, flipping two edges in the vertex $f_{ijk}$ gives a minus sign, see Figure \ref{figcolor2}.
So, the planar embedding determines the sign of the color factor.

\begin{figure}[ht]
	\begin{tikzpicture}[scale=1]
		\begin{scope}
		\coordinate (II) at (-1,0);
		\coordinate (JJ) at (1,0);
		\coordinate (KK) at (0,1.732);
		\coordinate (XX) at (0,0.577);
		\draw[fill = black] (XX) circle (1.5pt);
		\draw[ad](II) -- (XX);
		\draw[ad](JJ) -- (XX);
		\draw[ad](KK) -- (XX);
		\node [above=.01 of II,scale=1.1] {\textit{j}};
		\node [above=.01 of JJ,scale=1.1] {\textit{k}};
		\node [right=.01 of KK,scale=1.1] {\textit{i}};
		\end{scope}
		
		\begin{scope}[xshift=100]
		\coordinate (II) at (-1,0);
		\coordinate (JJ) at (1,0);
		\coordinate (KK) at (0,1.732);
		\coordinate (XX) at (0,0.577);
		\coordinate (A1) at (.9,0.45);
		\coordinate (A2) at (-.9,0.45);
		\draw [ad] (XX) .. controls (A1) .. (II);
		\draw [white,line width=5pt] (XX) .. controls (A2) .. (JJ);
		\draw[fill = black] (XX) circle (1.5pt);
		\draw[ad](KK) -- (XX);
		\draw [ad] (XX) .. controls (A2) .. (JJ);
		\node [above=.01 of II,scale=1.1] {\textit{j}};
		\node [above=.01 of JJ,scale=1.1] {\textit{k}};
		\node [right=.01 of KK,scale=1.1] {\textit{i}};
		\node [left=1.1 of XX,scale=1.1] {\text{$=\;-$}};
		\end{scope}
		
	\end{tikzpicture}
	\caption{Flipping two edges at an adjoint vertex gives a minus sign.}
	\label{figcolor2}
\end{figure}
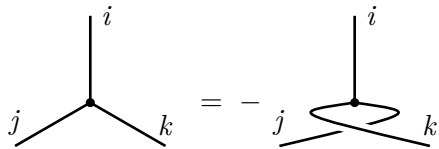

The graphical representations of (\ref{color1}) to (\ref{color5}) are in Figure \ref{figcolor3}, where we swapped the sides of (\ref{color1}).
To simplify the graphical notation it is customary to drop the labels of external half-edges using the convention that half-edges located at the
same position have equal labels.

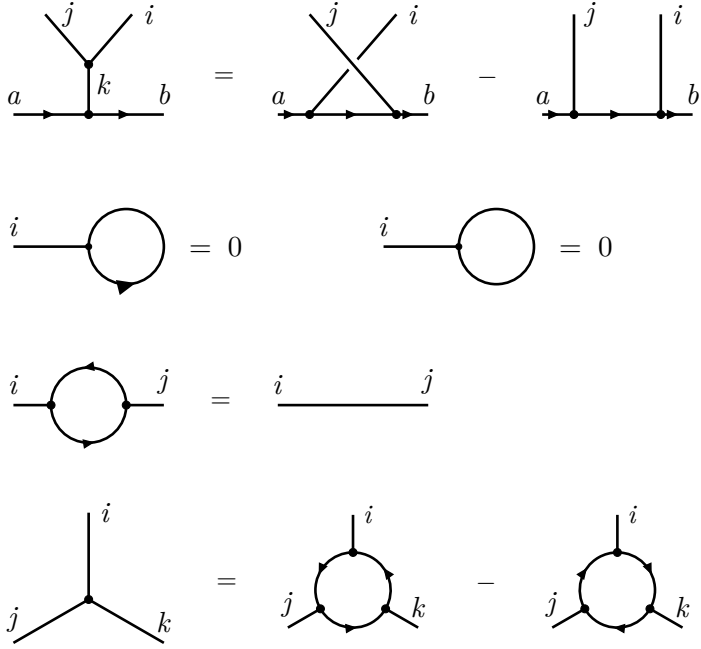
\begin{figure}[ht]
	\begin{tikzpicture}[baseline=(current bounding box.center)]
	\begin{scope}[local bounding box=vertex]
	\coordinate (AA) at (-1,0);
	\coordinate (BB) at (1,0);
	\coordinate (XX) at (0,0);
	\coordinate (YY) at (0,.666);
	\coordinate (JJ) at (-.577,1.333);
	\coordinate (II) at (.577,1.333);
	\draw[fill = black] (XX) circle (1.5pt);
	\draw[fill = black] (YY) circle (1.5pt);
	\draw[fun](AA) -- (XX);
	\draw[fun](XX) -- (BB);
	\draw[ad](XX) -- (YY);
	\draw[ad](JJ) -- (YY);
	\draw[ad](II) -- (YY);
	\node [above=.01 of AA,scale=1.1] {\textit{a}};
	\node [above=.01 of BB,scale=1.1] {\textit{b}};
	\node [scale=1.1] at (0.2,.4) {\textit{k}};
	\node [right=.1 of JJ,scale=1.1] {\textit{j}};
	\node [right=.01 of II,scale=1.1] {\textit{i}};
	\end{scope}
	
	\node at (1.8,0.5) {$=$};
	
	\begin{scope}[xshift=100,local bounding box=cross]
	\coordinate (AA) at (-1,0);
	\coordinate (BB) at (,0);
	\coordinate (XX) at (-.577,0);
	\coordinate (YY) at (.577,0);
	\coordinate (JJ) at (-.577,1.333);
	\coordinate (II) at (.577,1.333);
	\draw[fun](AA) -- (XX);
	\draw[ad](XX) -- (II);
	\draw [white,line width=5pt] (YY) -- (JJ);
	\draw[ad](YY) -- (JJ);
	\draw[fill = black] (XX) circle (1.5pt);
	\draw[fill = black] (YY) circle (1.5pt);
	\draw[fun](XX) -- (YY);
	\draw[fun](YY) -- (BB);
	\node [above=.01 of AA,scale=1.1] {\textit{a}};
	\node [above=.01 of BB,scale=1.1] {\textit{b}};
	\node [right=.1 of JJ,scale=1.1] {\textit{j}};
	\node [right=.01 of II,scale=1.1] {\textit{i}};	\end{scope}
	
	\node at (5.3,0.5) {$-$};
	
	\begin{scope}[xshift=200,local bounding box=parallel]
	\coordinate (AA) at (-1,0);
	\coordinate (BB) at (,0);
	\coordinate (XX) at (-.577,0);
	\coordinate (YY) at (.577,0);
	\coordinate (JJ) at (-.577,1.333);
	\coordinate (II) at (.577,1.333);
	\draw[fill = black] (XX) circle (1.5pt);
	\draw[fill = black] (YY) circle (1.5pt);
	\draw[fun](AA) -- (XX);
	\draw[fun](XX) -- (YY);
	\draw[fun](YY) -- (BB);
	\draw[ad](XX) -- (JJ);
	\draw[ad](YY) -- (II);
	\node [above=.01 of AA,scale=1.1] {\textit{a}};
	\node [above=.01 of BB,scale=1.1] {\textit{b}};
	\node [right=.01 of JJ,scale=1.1] {\textit{j}};
	\node [right=.01 of II,scale=1.1] {\textit{i}};
	\end{scope}
	
	\begin{scope}[local bounding box = funpole,
				  decoration={	markings,
						mark=at position .79 with {\arrow[scale=.8,>=triangle 45]{>}}
							},
				  yshift=-50]
	\coordinate (II) at (-1,0);
	\coordinate (CENTER) at (.5,0);
	\coordinate (XX) at (0,0);
	\draw[line width=1pt,postaction={decorate}] (CENTER) circle (.5 cm);
	\draw[ad](II) -- (XX);
	\draw[fill = black] (XX) circle (0.04 cm);
	\node [above=.01 of II,scale=1.1]{\textit{i}};
	\node [right=1.2 of XX,scale=1.1]{$=\;0$};
	\end{scope}

	\begin{scope}[xshift=140,yshift=-50]
	\coordinate (II) at (-1,0);
	\coordinate (CENTER) at (.5,0);
	\coordinate (XX) at (0,0);
	\draw[line width=1pt,postaction={decorate}] (CENTER) circle (.5 cm);
	\draw[ad](II) -- (XX);
	\draw[fill = black] (XX) circle (0.04 cm);
	\node [above=.01 of II,scale=1.1]{\textit{i}};
	\node [right=1.2 of XX,scale=1.1]{$=\;0$};
	\end{scope}
	
	\begin{scope}[yshift=-110, local bounding box=loop1]
	\coordinate (II) at (-1,0);
	\coordinate (JJ) at (1,0);
	\coordinate (XX) at (-.5,0);
	\coordinate (YY) at (.5,0);
	\draw[fun] (YY) arc (0:180:.5);
	\draw[fun] (XX) arc (180:360:.5);
	\draw[ad](II) -- (XX);
	\draw[ad](YY) -- (JJ);
	\draw[fill = black] (XX) circle (1.5pt);
	\draw[fill = black] (YY) circle (1.5pt);
	\node [above=.01 of II,scale=1.1] {\textit{i}};
	\node [above=.01 of JJ,scale=1.1] {\textit{j}};
	\end{scope}
	
	\begin{scope}[yshift=-110,xshift=100,local bounding box=line]
	\coordinate (II) at (-1,0);
	\coordinate (JJ) at (1,0);
	\draw[ad](II) -- (JJ);
	\node [above=.01 of II,scale=1.1] {\textit{i}};
	\node [above=.01 of JJ,scale=1.1] {\textit{j}};
	\end{scope}
	
	\path(loop1) -- (loop1-|line.west)  node[midway]{$=$};
	
	\begin{scope}[yshift=-200,local bounding box=vertex]
	\coordinate (JJ) at (-1,0);
	\coordinate (KK) at (1,0);
	\coordinate (II) at (0,1.732);
	\coordinate (XX) at (0,0.577);
	\draw[fill = black] (XX) circle (1.5pt);
	\draw[ad](JJ) -- (XX);
	\draw[ad](KK) -- (XX);
	\draw[ad](II) -- (XX);
	\node [above=.01 of JJ,scale=1.1] {\textit{j}};
	\node [above=.01 of KK,scale=1.1] {\textit{k}};
	\node [right=.01 of II,scale=1.1] {\textit{i}};
	\end{scope}
	
	\begin{scope}[xshift=100,yshift=-180,local bounding box=cclock]
	\coordinate (II) at (0,1);
	\coordinate (JJ) at (-.866,-.5);
	\coordinate (KK) at (.866,-.5);
	\coordinate (XX) at (0,.5);
	\coordinate (YY) at (-.433,-.25);
	\coordinate (ZZ) at (.433,-.25);
	\draw[fun] (XX) arc (90:210:.5);
	\draw[fun] (YY) arc (210:330:.5);
	\draw[fun] (ZZ) arc (-30:90:.5);
	\draw[ad](II) -- (XX);
	\draw[ad](JJ) -- (YY);
	\draw[ad](KK) -- (ZZ);
	\draw[fill = black] (XX) circle (1.5pt);
	\draw[fill = black] (YY) circle (1.5pt);
	\draw[fill = black] (ZZ) circle (1.5pt);
	\node [right=.01 of II,scale=1.] {\textit{i}};
	\node [above=.01 of JJ,scale=1.1] {\textit{j}};
	\node [above=.01 of KK,scale=1.1] {\textit{k}};
	\end{scope}
	
	\begin{scope}[xshift=200,yshift=-180,local bounding box =clock]
	\coordinate (II) at (0,1);
	\coordinate (JJ) at (-.866,-.5);
	\coordinate (KK) at (.866,-.5);
	\coordinate (XX) at (0,.5);
	\coordinate (YY) at (-.433,-.25);
	\coordinate (ZZ) at (.433,-.25);
	\draw[fun] (XX) arc (90:-30:.5);
	\draw[fun] (ZZ) arc (-30:-150:.5);
	\draw[fun] (YY) arc (-150:-270:.5);
	\draw[ad](II) -- (XX);
	\draw[ad](JJ) -- (YY);
	\draw[ad](KK) -- (ZZ);
	\draw[fill = black] (XX) circle (1.5pt);
	\draw[fill = black] (YY) circle (1.5pt);
	\draw[fill = black] (ZZ) circle (1.5pt);
	\node [right=.01 of II,scale=1.] {\textit{i}};
	\node [above=.01 of JJ,scale=1.1] {\textit{j}};
	\node [above=.01 of KK,scale=1.1] {\textit{k}};
	\end{scope}

	\path(vertex) -- (vertex-|cclock.west)  node[midway,below]{$=$};
	\path(cclock) -- (cclock-|clock.west)  node[midway,below]{$-$};

	\end{tikzpicture}
	\caption{Basic identities for color graphs.}
	\label{figcolor3}
\end{figure}

\begin{lem}\label{colorlem1}
An orthonormal basis of the fundamental (defining) representation of $sl(N)$ are the $N\times N$ matrices
\begin{align}\label{color6}
T^{\alpha\beta}&=\frac1{\sqrt2}(E_{\alpha\beta}+E_{\beta\alpha}),\qquad\tilde T^{\alpha\beta}=\frac\ii{\sqrt2}(E_{\alpha\beta}-E_{\beta\alpha})\quad\text{for }1\leq\alpha<\beta\leq N,\nonumber\\
T^k&=\frac1{\sqrt{k(k+1)}}\bigg(\Big(\sum_{r=1}^kE_{rr}\Big)-kE_{(k+1)(k+1)}\bigg)\quad\text{for }k=1,\ldots,N-1,
\end{align}
where $(E_{ab})_{cd}=\delta_{a,c}\delta_{b,d}$ are the elementary matrices.
We fix any sequence of $\alpha\beta$ in $T^{\alpha\beta}$ and $\tilde T^{\alpha\beta}$ to continue the labels $1,\ldots,N-1$ of $T^k$ to $N-1+2N(N-1)/2=N^2-1$.
\end{lem}
\begin{proof}
The orhonormality of $T^{\alpha\beta}$ and $\tilde T^{\alpha\beta}$ is clear.
Because for $\alpha\neq\beta$ and any $\gamma$ we have (no sum) $\Tr E_{\gamma\gamma}E_{\alpha\beta}=\delta_{\alpha,\gamma}\Tr E_{\alpha\beta}=0$,
the $T^{\alpha\beta}$ and $\tilde T^{\alpha\beta}$ are
orthogonal to the $T^k$. From
$$
(T^k)^2=\frac1{k(k+1)}\bigg(\Big(\sum_{r=1}^kE_{rr}\Big)+k^2E_{(k+1)(k+1)}\bigg)
$$
we obtain $\Tr (T^k)^2=1$.
\end{proof}
\begin{defn}
Let $G$ be a ribbon graph with oriented and non-oriented edges.
We assume that $G$ has two types of vertices: a (fundamental) vertex with two oriented edges and one non-oriented edge
and an (adjoint) vertex with three non-oriented edges.
We write $e_G$, $v_G$, $h_G$, $c_G$ for the number of edges, the number of vertices,
the number of independent cycles (loops), and the number of components of $G$, respectively.
The empty (self-)loop $\circ$ (top graphs in Figure \ref{figcolor5}) has
$e_\circ=v_\circ=0$ and $h_\circ=c_\circ=1$.

The graph $G$ can have external (non-paired) half-edges (hairs).
Equivalently, we connect external vertices to the external
half-edges and obtain a decomposition of the vertices into external (one-valent) and internal (three-valent) vertices, $v_G=v_G^{\mathrm{ext}}+v_G^{\mathrm{int}}$.

After the removal of the non-oriented edges, $G$ decomposes into a collection of $f_G$ oriented cycles.

The reduction of the color graph $G$ is $R_G(N)$.
\end{defn}
From graph homology we get
\begin{equation}\label{graphhom}
h_G-e_G+v_G-c_G=0
\end{equation}
and from counting half-edges we obtain
\begin{equation}\label{halfedges}
v_G^{\mathrm{ext}}=2e_G-3v_G^{\mathrm{int}}.
\end{equation}

\subsection{Fundamental identities}
\begin{figure}[ht]
	\begin{tikzpicture}[scale=.85]
	\begin{scope}[local bounding box=eitch]
	\coordinate (AA) at (-1,0);
	\coordinate (BB) at (1,0);
	\coordinate (CC) at (-1,1.532);
	\coordinate (DD) at (1,1.532);
	\coordinate (XX) at (0,0);
	\coordinate (YY) at (0,1.532);
	\draw[fun](BB) -- (XX);
	\draw[fun](XX) -- (AA);
	\draw[ad](XX) -- (YY);
	\draw[fun](CC) -- (YY);
	\draw[fun](YY) -- (DD);
	\draw[fill = black] (XX) circle (1.5pt);
	\draw[fill = black] (YY) circle (1.5pt);
	\end{scope}

	\begin{scope}[xshift=100,local bounding box=hourglass]
	\coordinate (BB) at (.6,0);
	\coordinate (CC) at (-.6,1.532);
	\draw[fun] (CC) arc (50:-50:1);
	\draw[fun] (BB) arc (230:130:1);
	\end{scope}
	
	\begin{scope}[xshift=220,local bounding box=twolines]
	\coordinate (AA) at (-1,0);
	\coordinate (BB) at (1,0);
	\coordinate (CENT) at (-1,.766);
	\coordinate (CC) at (-1,1.532);
	\coordinate (DD) at (1,1.532);
	\draw[fun](BB) -- (AA);
	\draw[fun](CC) -- (DD);
	\node[left=.2 of CENT,scale=1.1]{\LARGE$\frac1N$};
	\end{scope}
	
	\path(eitch) -- (eitch-|hourglass.west)  node[midway,scale=1.1]{$=$};
	\path(hourglass) -- (hourglass-|twolines.west)  node[midway,scale=1.1]{$-$};
	
	\end{tikzpicture}
	\caption{The two-term relation of $SU(N)$ color reductions, Equation (\ref{color7}).}
	\label{figcolor4}
\end{figure}
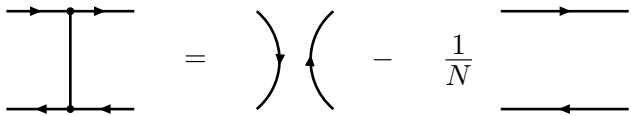
The core identity for color reductions is depicted in Figure \ref{figcolor4}.
This identity is specific to the Lie algebras $\slN$. Similar reduction formulae for other Lie algebras are in \cite{Cvit}.
\begin{prop}\label{prop:color}
In any orthonormal basis $T^i$, we obtain for the sum $\sum_iT^i\otimes T^i$ over all tensor squares of the basis elements,
\begin{equation}\label{color7}
T^i_{ab}T^i_{cd}=\delta_{a,d}\delta_{b,c}-\frac1N\delta_{a,b}\delta_{c,d}.
\end{equation}
\end{prop}
\begin{proof}
See \cite{Haber} for a short proof which does not use an explicit basis. Here, we use the basis in Lemma \ref{colorlem1} for an explicit calculation.

We first check (\ref{color7}) for the orthonormal basis in Lemma \ref{colorlem1}.
The sum over $\alpha\beta$ in $T^{\alpha\beta}\otimes T^{\alpha\beta}$ and $\tilde T^{\alpha\beta}\otimes\tilde T^{\alpha\beta}$ gives
$$
\frac12\sum_{1\leq\alpha<\beta\leq N}(\delta_{a,\alpha}\delta_{b,\beta}+\delta_{a,\beta}\delta_{b,\alpha})(\delta_{c,\alpha}\delta_{d,\beta}+\delta_{c,\beta}\delta_{d,\alpha})
-(\delta_{a,\alpha}\delta_{b,\beta}-\delta_{a,\beta}\delta_{b,\alpha})(\delta_{c,\alpha}\delta_{d,\beta}-\delta_{c,\beta}\delta_{d,\alpha}).
$$
Only the cross terms survive, yielding
$$
\sum_{1\leq\alpha<\beta\leq N}(\delta_{a,\alpha}\delta_{b,\beta}\delta_{c,\beta}\delta_{d,\alpha}+\delta_{a,\beta}\delta_{b,\alpha}\delta_{c,\alpha}\delta_{d,\beta})
=\delta_{a,d}\delta_{b,c}\sum_{1\leq\alpha<\beta\leq N}(\delta_{a,\alpha}\delta_{b,\beta}+\delta_{b,\alpha}\delta_{a,\beta}).
$$
The first term in the sum gives $1$ for $a<b$ and $0$ for $a\geq b$. The second term is $1$ for $b<a$ and $0$ for $b\geq a$. Adding and subtracting $\delta_{a,b}$ yields
$$
\delta_{a,d}\delta_{b,c}(1-\delta_{a,b}).
$$
We assume $a\leq c$ (otherwise we swap $a,b$ with $c,d$) and obtain for the sum over $T^k\otimes T^k$,
\begin{align*}
&\sum_{k=1}^{N-1}\frac1{k(k+1)}\bigg(\Big(\sum_{r=1}^k\delta_{a,r}\delta_{b,r}\Big)-k\delta_{a,k+1}\delta_{b,k+1}\bigg)\bigg(\Big(\sum_{s=1}^k\delta_{c,s}\delta_{d,s}\Big)-k\delta_{c,k+1}\delta_{d,k+1}\Big)\\
=&\delta_{a,b}\delta_{c,d}\sum_{k=1}^{N-1}\frac1{k(k+1)}\bigg(\Big(\sum_{r,s=1}^k\delta_{a,r}\delta_{c,s}\Big)-k\delta_{c,k+1}\Big(\sum_{r=1}^k\delta_{a,r}\Big)-k\delta_{a,k+1}\Big(\sum_{s=1}^k\delta_{c,s}\Big)+
k^2\delta_{a,k+1}\delta_{c,k+1}\bigg)\\
=&\delta_{a,b}\delta_{c,d}\Big(\sum_{k=c}^{N-1}\frac1{k(k+1)}-\frac1c(1-\delta_{a,c})-0+\delta_{a,c}\frac{c-1}c\Big).
\end{align*}
Because $1/(k(k+1))=1/k-1/(k+1)$, the sum over $k$ yields $1/c-1/N$. Altogether we get
$$
\delta_{a,b}\delta_{c,d}(-\frac1N+\delta_{a,c}).
$$
Adding the sum over $\alpha\beta$ and the sum over $k$, we see that the terms with three Kronecker deltas (implying $a=b=c=d$) cancel and (\ref{color7}) follows.

To get the result for any orthonormal basis we transform $T^i$ to $T'^i=ST^iS^{-1}$ for some invertible $n\times n$ matrix $S$. In components, we get for $\sum_iT'^i\otimes T'^i$
\begin{align*}
T'^i_{ab}T'^i_{cd}&=S_{aa'}T^i_{a'b'}(S^{-1})_{b'b}S_{cc'}T^i_{c'd'}(S^{-1})_{d'd}=S_{aa'}(S^{-1})_{b'b}S_{cc'}(S^{-1})_{d'd}(\delta_{a',d'}\delta_{b',c'}-\frac1N\delta_{a',b'}\delta_{c',d'})\\
&=S_{aa'}(S^{-1})_{a'd}S_{cb'}(S^{-1})_{b'b}-\frac1N S_{aa'}(S^{-1})_{a'b}S_{cc'}(S^{-1})_{c'd})=\delta_{a,d}\delta_{b,c}-\frac1N\delta_{a,b}\delta_{c,d}.
\end{align*}
This completes the proof of the proposition.
\end{proof}
Note that in Figure \ref{figcolor4} the orientation of the external half-edges is preserved. This promotes (\ref{color7}) to a two term relation which is the oriented analog of the classical tree term relation in graph theory (the cross term is forbidden because it is in conflict with the orientation of the edges).

With (\ref{color4}) we can eliminate all adjoint vertices $f_{ijk}$ and with (\ref{color7}) we can eliminate all non-oriented edges between two chains of oriented edges. For any graph with at least one vertex, we finally obtain a result which is a product of
(1) oriented chains with any number of non-oriented external half-edges (products of $T^i$) and
(2) closed oriented loops with any number of non-oriented external half-edges
(traces of products of $T^i$). In the case of a `vacuum' graph with no
external half-edges, we only get a sum of oriented self-loops $\delta_{a,a}=N$. Then, $R_G(N)\in\ZZ[N,N^{-1}]$.

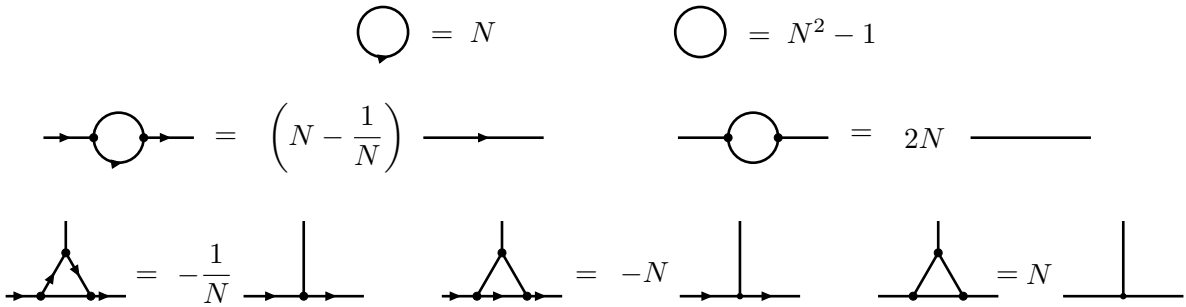
\begin{figure}
	\begin{tikzpicture}
	
	\begin{scope}[xshift=120,
	decoration={	markings,
		mark=at position .79 with {\arrow[scale=.5,>=triangle 45]{>}}
	}
	]
	\coordinate (CENTER) at (0,0);
	\draw[line width=1pt,postaction={decorate}] (CENTER) circle (.333);
	\node [right=.5 of CENTER,scale=1.1]{$=\;N$};
	\end{scope}
	
	\begin{scope}[xshift=240]
	\coordinate (CENTER) at (0,0);
	\draw[ad] (CENTER) circle (.333);
	\node [right=.5 of CENTER,scale=1.1]{$=\;N^2-1$};
	\end{scope}
	
	
		
	\begin{scope}[xshift=20,yshift=-40,local bounding box=loop1]
	\coordinate (II) at (-1,0);
	\coordinate (JJ) at (1,0);
	\coordinate (XX) at (-.333,0);
	\coordinate (YY) at (.333,0);
	\draw[ad] (YY) arc (0:180:.333);
	\draw[fun] (XX) arc (180:360:.333);
	\draw[fun](II) -- (XX);
	\draw[fun](YY) -- (JJ);
	\draw[fill = black] (XX) circle (1.5pt);
	\draw[fill = black] (YY) circle (1.5pt);
	\end{scope}
	
	\begin{scope}[xshift=158,yshift=-40,local bounding box=line1]
	\coordinate (AA) at (-.8,0);
	\coordinate (BB) at (.8,0);
	\draw[fun](AA) -- (BB);
	\node[left=1.7 of BB,scale=1.1]{$\left(N-\displaystyle\frac1N\right)$};
	\end{scope}
	
	\path(loop1) -- (loop1-|line1.west)  node[midway,scale=1.1]{$=$};
	
	
	

	\begin{scope}[xshift=260,yshift=-40,local bounding box=loop3]
	\coordinate (II) at (-1,0);
	\coordinate (JJ) at (1,0);
	\coordinate (XX) at (-.333,0);
	\coordinate (YY) at (.333,0);
	\draw[ad] (YY) arc (0:180:.333);
	\draw[ad] (XX) arc (180:360:.333);
	\draw[ad](II) -- (XX);
	\draw[ad](YY) -- (JJ);
	\draw[fill = black] (XX) circle (1.5pt);
	\draw[fill = black] (YY) circle (1.5pt);
	\end{scope}
	
	\begin{scope}[yshift=-40,xshift=365,local bounding box = line3]
	\coordinate (II) at (-.8,0);
	\coordinate (JJ) at (.8,0);
	\draw[ad](II) -- (JJ);
	\node[left=.2 of II,scale=1.1]{$2N$};
	\end{scope}
	
	\path(loop3) -- (loop3-|line3.west)  node[midway,below=-.25,scale=1.1]{$=$};
		
	\begin{scope}[yshift=-100,local bounding box=triangle1]
	\coordinate (AA) at (-.8,0);
	\coordinate (BB) at (.8,0);
	\coordinate (CC) at (0,1);
	\coordinate (XX) at (-.333,0);
	\coordinate (YY) at (.333,0);
	\coordinate (ZZ) at (0,.577);
	\draw[fun](AA) -- (XX);
	\draw[ad](XX) -- (YY);
	\draw[fun](YY) -- (BB);
	\draw[fun](XX) -- (ZZ);
	\draw[fun](ZZ) -- (YY);
	\draw[ad](ZZ) -- (CC);
	\draw[fill = black] (XX) circle (1.5pt);
	\draw[fill = black] (YY) circle (1.5pt);
	\draw[fill = black] (ZZ) circle (1.5pt);
	\end{scope}
	
	\begin{scope}[yshift=-100,xshift=90,local bounding box= fork1]
	\coordinate (AA) at (-.8,0);
	\coordinate (BB) at (.8,0);
	\coordinate (XX) at (0,0);
	\coordinate (CC) at (0,1);
	\coordinate (MID) at (-.8,.3);
	\draw[fun](AA) -- (XX);
	\draw[fun](XX) -- (BB);
	\draw[ad](CC) -- (XX);
	\draw[fill = black] (XX) circle (1.5pt);
	\node[left=.0 of MID,scale=1.1]{$-\displaystyle\frac1N$};
	\end{scope}
	
	\path(triangle1) -- (triangle1-|fork1.west)  node[midway,below,scale=1.1]{$=$};
	
	\begin{scope}[yshift=-100,xshift=165,local bounding box=triangle2]
	\coordinate (AA) at (-.8,0);
	\coordinate (BB) at (.8,0);
	\coordinate (CC) at (0,1);
	\coordinate (XX) at (-.333,0);
	\coordinate (YY) at (.333,0);
	\coordinate (ZZ) at (0,.577);
	\draw[fun](AA) -- (XX);
	\draw[fun](XX) -- (YY);
	\draw[fun](YY) -- (BB);
	\draw[ad](XX) -- (ZZ);
	\draw[ad](ZZ) -- (YY);
	\draw[ad](ZZ) -- (CC);
	\draw[fill = black] (XX) circle (1.5pt);
	\draw[fill = black] (YY) circle (1.5pt);
	\draw[fill = black] (ZZ) circle (1.5pt);
	\end{scope}
	
	\begin{scope}[yshift=-100,xshift=255,local bounding box= fork2]
	\coordinate (AA) at (-.8,0);
	\coordinate (BB) at (.8,0);
	\coordinate (XX) at (0,0);
	\coordinate (CC) at (0,1);
	\coordinate (MID) at (-.8,.3);
	\draw[fun](AA) -- (XX);
	\draw[fun](XX) -- (BB);
	\draw[ad](CC) -- (XX);
	\draw[fill = black] (XX) circle (1pt);
	\node[left=.0 of MID,scale=1.1]{$-N$};
	\end{scope}
	
	\path(triangle2) -- (triangle2-|fork2.west)  node[midway,below,scale=1.1]{$=$};
	
	\begin{scope}[xshift=330,yshift=-100,local bounding box=triangle3]
	\coordinate (AA) at (-.8,0);
	\coordinate (BB) at (.8,0);
	\coordinate (CC) at (0,1);
	\coordinate (XX) at (-.333,0);
	\coordinate (YY) at (.333,0);
	\coordinate (ZZ) at (0,.577);
	\draw[ad](AA) -- (XX);
	\draw[ad](XX) -- (YY);
	\draw[ad](YY) -- (BB);
	\draw[ad](XX) -- (ZZ);
	\draw[ad](ZZ) -- (YY);
	\draw[ad](ZZ) -- (CC);
	\draw[fill = black] (XX) circle (1.5pt);
	\draw[fill = black] (YY) circle (1.5pt);
	\draw[fill = black] (ZZ) circle (1.5pt);
	\end{scope}
	
	\begin{scope}[yshift=-100,xshift=400,local bounding box= fork3]
	\coordinate (AA) at (-.8,0);
	\coordinate (BB) at (.8,0);
	\coordinate (XX) at (0,0);
	\coordinate (CC) at (0,1);
	\coordinate (MID) at (-.8,.3);
	\draw[ad](AA) -- (XX);
	\draw[ad](XX) -- (BB);
	\draw[ad](CC) -- (XX);
	\draw[fill = black] (XX) circle (1pt);
	\node[left=.0 of MID,scale=1.1]{$N$};
	\end{scope}
	
	\path(triangle3) -- (triangle3-|fork3.west)  node[midway,below,scale=1.1]{$=$};
	
	\end{tikzpicture}
	\caption{Reductions of small cycles (also see Figure \ref{figcolor3}).}
	\label{figcolor5}
\end{figure}

For practical purposes, it is useful to derive formulae for small cycles. We have (see Figures \ref{figcolor3} and \ref{figcolor5})
\begin{align}\label{color8}
\delta_{a,a}=N,&\qquad \delta_{i,i}=N^2-1,\nonumber\\
T^i_{aa}=0,&\qquad f_{ijj}=0,\nonumber\\
T^i_{ab}T^i_{bc}=(N-\frac1N)\delta_{a,c},&\qquad T^i_{ab}T^j_{ba}=\delta_{i,j},\qquad f_{ik\ell}f_{\ell kj}=2N\delta_{i,j},\nonumber\\
T^i_{ab}T^j_{bc}T^i_{cd}=-\frac1N T^j_{ad},&\qquad T^i_{ab}T^j_{bc}f_{ijk}=-NT^k_{ac},\qquad f_{ijk}f_{j\ell m}f_{mnk}=Nf_{i\ell n}.
\end{align}
Oriented cycles with $\geq3$ edges cannot be reduced.

\subsection{The reduction algorithm}
An efficient algorithm for color reduction first searches for the smallest cycle in the graph. If the number of vertices in this cycle (the girth of the graph) is $\leq3$ and the cycle is not
an oriented triangle, then we use (\ref{color8}) for reduction.

If the girth of the graph is $\geq3$ and all triangles are oriented, then we search for a non-oriented edge that connects two oriented chains and use
(\ref{color7}). The graph loses two fundamental vertices. The number of independent cycles decreases by one or the graph disconnects. The color factors of disconnected graphs factorize.

If none of the previous reductions is possible, then we search for the smallest non-oriented cycle (it still may have oriented edges).
If the cycle has an oriented edge, then we use the first identity in Figure \ref{figcolor3} to reduce an adjoint vertex in the cycle. The number of adjoint vertices in the cycle
decreases by one. In one term the cycle also shrinks by an edge.
We go back to the elimination of non-oriented edges between oriented chains.

If all smallest non-oriented cycles have only non-oriented edges, then we use (\ref{color4}) to convert one adjoint vertex in the cycle into a sum of two oriented triangles
(also see Proposition \ref{prop:results}).
The number of vertices in the original cycle (and also $h_G$) increases by one. But now the cycle has an oriented edge and can be reduced.

The algorithm terminates if every graph in the reduction is a union of oriented chains and cycles,
possibly with external non-oriented half-edges.
For a connected vacuum graph with at least one vertex, every term in the reduction
is a collection of oriented self-loops.

The algorithm tries (whenever possible) to avoid producing a large number of terms by eliminating adjoint vertices with (\ref{color4}). It also quickly reduces $h_G$ so
that one obtains a bootstrap algorithm. This is particularly powerful in the case of vacuum graphs (graphs with no external half-edges). For such
graphs, the result is in $\ZZ[N,N^{-1}]$ and can be cached for small $h_G$. Also, for a given $h_G$ the number of vacuum graphs is much smaller
than the number of graphs with external half-edges. This double effect (simple results and few graphs) allows one to cache all vacuum graphs with $h_G\leq11$ which are not amenable
to reductions in (\ref{color8}). This algorithm has been implemented in HyperlogProcedures \cite{Shlog}.

\begin{ex}\label{colorex2}
Nontrivial vacuum graphs with four loops and their reductions are in Figure \ref{figcolor7}.
\end{ex}

\begin{figure}[ht]
	\begin{tikzpicture}
	\begin{scope}
	\coordinate (AA) at (-1,1);
	\coordinate (BB) at (-1,0);
	\coordinate (CC) at (-1,-1);
	\coordinate (XX) at (1,1);
	\coordinate (YY) at (1,0);
	\coordinate (ZZ) at (1,-1);
	\coordinate (MID) at (0,-1);
	\draw[ad](AA) -- (XX);
	\draw[ad](AA) -- (YY);
	\draw[ad](AA) -- (ZZ);
	\draw[white,line width=4pt] (BB) -- (XX);
	\draw[ad](BB) -- (XX);
	\draw[white,line width=4pt] (BB) -- (YY);
	\draw[ad](BB) -- (YY);
	\draw[ad](BB) -- (ZZ);
	\draw[white,line width=4pt] (CC) -- (XX);
	\draw[ad](CC) -- (XX);
	\draw[white,line width=4pt] (CC) -- (YY);
	\draw[ad](CC) -- (YY);
	\draw[ad](CC) -- (ZZ);
	\draw[fill = black] (AA) circle (1.5pt);
	\draw[fill = black] (BB) circle (1.5pt);
	\draw[fill = black] (CC) circle (1.5pt);
	\draw[fill = black] (XX) circle (1.5pt);
	\draw[fill = black] (YY) circle (1.5pt);
	\draw[fill = black] (ZZ) circle (1.5pt);
	\node[below=.2 of MID,scale=1.1]{\Large$0$};
	\end{scope}
	
	\begin{scope}[xshift=150]
	\coordinate (WW) at (0,1);
	\coordinate (XX) at (-1,0);
	\coordinate (YY) at (0,-1);
	\coordinate (ZZ) at (1,0);
	\coordinate (AA) at (-.354,.354);
	\coordinate (BB) at (.354,-.354);
	\coordinate (MID) at (0,-1);
	\draw[fun] (WW) arc (90:180:1);
	\draw[fun] (XX) arc (180:270:1);
	\draw[fun] (YY) arc (-90:0:1);
	\draw[fun] (ZZ) arc (0:90:1);
	\draw[ad](WW) -- (AA);
	\draw[ad](XX) -- (BB);
	\draw [white,line width=4pt] (YY) -- (AA);
	\draw[ad](YY) -- (AA);
	\draw[ad](ZZ) -- (BB);
	\draw[ad](AA) -- (BB);
	\draw[fun] (WW) arc (90:180:1);
	\draw[fun] (XX) arc (180:270:1);
	\draw[fun] (YY) arc (-90:0:1);
	\draw[fun] (ZZ) arc (0:90:1);
	\draw[fill = black] (WW) circle (1.5pt);
	\draw[fill = black] (XX) circle (1.5pt);
	\draw[fill = black] (YY) circle (1.5pt);
	\draw[fill = black] (ZZ) circle (1.5pt);
	\draw[fill = black] (AA) circle (1.5pt);
	\draw[fill = black] (BB) circle (1.5pt);
	\node[below=.2 of MID,scale=1.1]{\Large$0$};
	\end{scope}
	
	\begin{scope}[xshift=300]
	\coordinate (WW) at (0,1);
	\coordinate (XX) at (-1,0);
	\coordinate (YY) at (0,-1);
	\coordinate (ZZ) at (1,0);
	\coordinate (AA) at (-.354,.354);
	\coordinate (BB) at (.354,-.354);
	\coordinate (MID) at (0,-1);
	\draw[ad](WW) -- (AA);
	\draw[ad](XX) -- (BB);
	\draw [white,line width=3mm] (YY) -- (AA);
	\draw[ad](YY) -- (AA);
	\draw[ad](ZZ) -- (BB);
	\draw[ad](AA) -- (BB);
	\draw[fun] (WW) arc (90:0:1);
	\draw[fun] (ZZ) arc (0:-90:1);
	\draw[fun] (YY) arc (-90:-180:1);
	\draw[fun] (XX) arc (-180:-270:1);
	\draw[fill = black] (WW) circle (1.5pt);
	\draw[fill = black] (XX) circle (1.5pt);
	\draw[fill = black] (YY) circle (1.5pt);
	\draw[fill = black] (ZZ) circle (1.5pt);
	\draw[fill = black] (AA) circle (1.5pt);
	\draw[fill = black] (BB) circle (1.5pt);
	\node[below=.2 of MID,scale=1.1]{\Large$0$};
	\end{scope}
	
	\begin{scope}[yshift=-100]
	\coordinate (UU) at (0,1);
	\coordinate (VV) at (-.866,.5);
	\coordinate (WW) at (-.866,-.5);
	\coordinate (XX) at (0,-1);
	\coordinate (YY) at (.866,-.5);
	\coordinate (ZZ) at (.866,.5);
	\coordinate (MID) at (0,-1);
	\draw[ad](UU) -- (XX);
	\draw [white,line width=4pt] (VV) -- (YY);
	\draw[ad](VV) -- (YY);
	\draw [white,line width=4pt] (WW) -- (ZZ);
	\draw[ad](WW) -- (ZZ);
	\draw[fun] (UU) arc (90:150:1);
	\draw[fun] (VV) arc (150:210:1);
	\draw[fun] (WW) arc (210:270:1);
	\draw[fun] (XX) arc (270:330:1);
	\draw[fun] (YY) arc (-30:30:1);
	\draw[fun] (ZZ) arc (30:90:1);
	\draw[fill = black] (UU) circle (1.5pt);
	\draw[fill = black] (VV) circle (1.5pt);
	\draw[fill = black] (WW) circle (1.5pt);
	\draw[fill = black] (XX) circle (1.5pt);
	\draw[fill = black] (YY) circle (1.5pt);
	\draw[fill = black] (ZZ) circle (1.5pt);
	\node[below=.2 of MID,scale=1.1]{$(1+\frac{1}{N^2})(N^2-1)$};
	\end{scope}
	
	\begin{scope}[yshift=-100,xshift=150]
	\coordinate (XX) at (0,1);
	\coordinate (YY) at (-.866,-.5);
	\coordinate (ZZ) at (.866,-.5);
	\coordinate (AA) at (0,.333);
	\coordinate (BB) at (-.289,-.167);
	\coordinate (CC) at (.289,-.167);
	\coordinate (MID) at (0,-1);
	\draw[fun] (XX) arc (90:210:1);
	\draw[fun] (YY) arc (210:330:1);
	\draw[fun] (ZZ) arc (-30:90:1);
	\draw[fun] (AA) arc (90:210:.333);
	\draw[fun] (BB) arc (210:330:.333);
	\draw[fun] (CC) arc (-30:90:.333);
	\draw[ad](XX) -- (AA);
	\draw[ad](YY) -- (BB);
	\draw[ad](ZZ) -- (CC);
	\draw[fill = black] (AA) circle (1.5pt);
	\draw[fill = black] (BB) circle (1.5pt);
	\draw[fill = black] (CC) circle (1.5pt);
	\draw[fill = black] (XX) circle (1.5pt);
	\draw[fill = black] (YY) circle (1.5pt);
	\draw[fill = black] (ZZ) circle (1.5pt);
	\node[below=.2 of MID,scale=1.1]{$-\frac2N(N^2-1)$};
	\end{scope}
	
	\begin{scope}[yshift=-100,xshift=300]
	\coordinate (XX) at (0,1);
	\coordinate (YY) at (-.866,-.5);
	\coordinate (ZZ) at (.866,-.5);
	\coordinate (AA) at (0,.333);
	\coordinate (BB) at (-.289,-.167);
	\coordinate (CC) at (.289,-.167);
	\coordinate (MID) at (0,-1);
	\draw[fun] (XX) arc (90:210:1);
	\draw[fun] (YY) arc (210:330:1);
	\draw[fun] (ZZ) arc (-30:90:1);
	\draw[fun] (AA) arc (90:-30:.333);
	\draw[fun] (CC) arc (-30:-150:.333);
	\draw[fun] (BB) arc (-150:-270:.333);
	\draw[ad](XX) -- (AA);
	\draw[ad](YY) -- (BB);
	\draw[ad](ZZ) -- (CC);
	\draw[fill = black] (AA) circle (1.5pt);
	\draw[fill = black] (BB) circle (1.5pt);
	\draw[fill = black] (CC) circle (1.5pt);
	\draw[fill = black] (XX) circle (1.5pt);
	\draw[fill = black] (YY) circle (1.5pt);
	\draw[fill = black] (ZZ) circle (1.5pt);
	\node[below=.2 of MID,scale=1.1]{$(N-\frac2N)(N^2-1)$};
	\end{scope}
	\end{tikzpicture}
	\caption{Nontrivial color graphs with four loops and their reductions.}
	\label{figcolor7}
\end{figure}
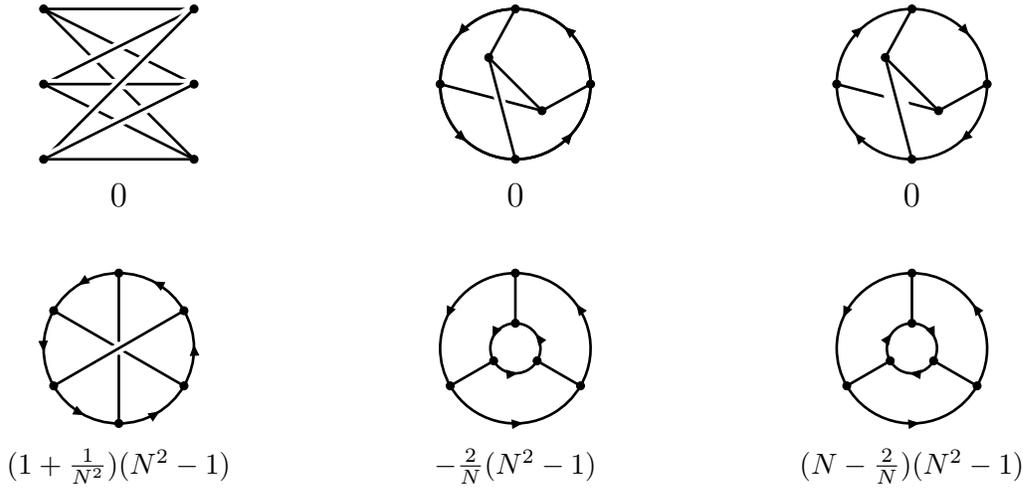

The color reduction of any graph with external half-edges can be represented as a sum of vacuum graphs by completing the graph in various ways,
see Example \ref{colorex3}. This gives a linear system which can easily be solved for a small number of external half-edges. Using completion
for many external half-edges has the drawback that $h_G$ increases and the method becomes less powerful.
In this case, direct reduction is more efficient.

\begin{figure}[ht]
	\begin{tikzpicture}
	\begin{scope}[local bounding box=graph1]
	\coordinate (II) at (-1,0);
	\coordinate (JJ) at (1,0);
	\coordinate (XX) at (-.5,0);
	\coordinate (YY) at (.5,0);
	\draw[line width = 1pt,fill=lightgray] (0,0) circle (.5 cm);
	\draw[ad](II) -- (XX);
	\draw[ad](YY) -- (JJ);
	\draw[fill = black] (XX) circle (1.5pt);
	\draw[fill = black] (YY) circle (1.5pt);
	\end{scope}
	
	\begin{scope}[yshift=-80,local bounding box=graph2]
	\coordinate (II) at (-1,0);
	\coordinate (JJ) at (1,0);
	\coordinate (XX) at (-.5,0);
	\coordinate (YY) at (.5,0);
	\coordinate (A1) at (2,1);
	\coordinate (A2) at (-2,1);
	\coordinate (MID) at (0,-1);
	\draw[line width = 1pt,fill=lightgray] (0,0) circle (.5 cm);
	\draw[ad] (YY) .. controls (A1) and (A2) .. (XX);
	\draw[fill = black] (XX) circle (1.5pt);
	\draw[fill = black] (YY) circle (1.5pt);
	\node[below=-.2 of MID,scale=1.1]{$G_0$};
	\end{scope}
	
	\draw[-latex,thick]([yshift=-5]graph1.south) -- ([yshift=5]graph1.south|-graph2.north);

	\begin{scope}[xshift=95,local bounding box=graph3]
	\coordinate (II) at (-1,0);
	\coordinate (JJ) at (1,0);
	\coordinate (XX) at (-.5,0);
	\coordinate (YY) at (.5,0);
	\draw[line width = 1pt,fill=lightgray] (0,0) circle (.5 cm);
	\draw[fun](II) -- (XX);
	\draw[fun](YY) -- (JJ);
	\draw[fill = black] (XX) circle (1.5pt);
	\draw[fill = black] (YY) circle (1.5pt);
	\end{scope}
	
	\begin{scope}[xshift=95,yshift=-80,local bounding box=graph4]
	\coordinate (II) at (-1,0);
	\coordinate (JJ) at (1,0);
	\coordinate (XX) at (-.5,0);
	\coordinate (YY) at (.5,0);
	\coordinate (A1) at (2,1);
	\coordinate (A2) at (-2,1);
	\coordinate (MID) at (0,-1);
	\draw[line width = 1pt,fill=lightgray] (0,0) circle (.5 cm);
	\draw[fun] (YY) .. controls (A1) and (A2) .. (XX);
	\draw[fill = black] (XX) circle (1.5pt);
	\draw[fill = black] (YY) circle (1.5pt);
	\node[below=-.2 of MID,scale=1.1]{$G_0$};
	\end{scope}
	
	\draw[-latex,thick]([yshift=-5]graph3.south) -- ([yshift=5]graph3.south|-graph4.north);

	\begin{scope}[xshift=190,local bounding box=graph5]
	\coordinate (II) at (-1,0);
	\coordinate (JJ) at (1,0);
	\coordinate (AA) at (0,1);
	\coordinate (XX) at (-.5,0);
	\coordinate (YY) at (.5,0);
	\coordinate (ZZ) at (0,.5);
	\draw[line width = 1pt,fill=lightgray] (0,0) circle (.5 cm);
	\draw[fun](II) -- (XX);
	\draw[fun](YY) -- (JJ);
	\draw[ad](ZZ) -- (AA);
	\draw[fill = black] (XX) circle (1.5pt);
	\draw[fill = black] (YY) circle (1.5pt);
	\draw[fill = black] (ZZ) circle (1.5pt);
	\end{scope}
	
	\begin{scope}[yshift=-80,xshift=190,local bounding box=graph6]
	\coordinate (II) at (-1,0);
	\coordinate (JJ) at (1,0);
	\coordinate (AA) at (0,1);
	\coordinate (XX) at (-.5,0);
	\coordinate (YY) at (.5,0);
	\coordinate (ZZ) at (0,.5);
	\coordinate (A1) at (1.2,0);
	\coordinate (A2) at (1.2,1);
	\coordinate (A3) at (-1.2,0);
	\coordinate (A4) at (-1.2,1);
	\coordinate (MID) at (0,-1);
	\draw[line width = 1pt,fill=lightgray] (0,0) circle (.5 cm);
	\draw[fun] (YY) .. controls (A1) and (A2) .. (AA);
	\draw[fun] (AA) .. controls (A4) and (A3) .. (XX);
	\draw[ad](ZZ) -- (AA);
	\draw[fill = black] (XX) circle (1.5pt);
	\draw[fill = black] (YY) circle (1.5pt);
	\draw[fill = black] (ZZ) circle (1.5pt);
	\node[below=-.2 of MID,scale=1.1]{$G_0$};
	\end{scope}
	
	\draw[-latex,thick]([yshift=-5]graph5.south) -- ([yshift=5]graph5.south|-graph6.north);

	\begin{scope}[xshift=320,local bounding box=graph7]
	\coordinate (II) at (-1,0);
	\coordinate (JJ) at (1,0);
	\coordinate (AA) at (0,1);
	\coordinate (XX) at (-.5,0);
	\coordinate (YY) at (.5,0);
	\coordinate (ZZ) at (0,.5);
	\draw[line width = 1pt,fill=lightgray] (0,0) circle (.5 cm);
	\draw[ad](II) -- (XX);
	\draw[ad](YY) -- (JJ);
	\draw[ad](ZZ) -- (AA);
	\draw[fill = black] (XX) circle (1.5pt);
	\draw[fill = black] (YY) circle (1.5pt);
	\draw[fill = black] (ZZ) circle (1.5pt);
	\end{scope}
	
	\begin{scope}[yshift=-70,xshift=280,local bounding box=graph8]
	\coordinate (II) at (-1,0);
	\coordinate (JJ) at (1,0);
	\coordinate (AA) at (0,1);
	\coordinate (XX) at (-.5,0);
	\coordinate (YY) at (.5,0);
	\coordinate (ZZ) at (0,.5);
	\coordinate (MID) at (0,-1);
	\draw[line width = 1pt,fill=lightgray] (0,0) circle (.5 cm);
	\draw[ad](II) -- (XX);
	\draw[ad](YY) -- (JJ);
	\draw[ad](ZZ) -- (AA);
	\draw[fun](AA) arc (90:180:1);
	\draw[fun](II) arc (-180:0:1);
	\draw[fun](JJ) arc (0:90:1);
	\draw[fill = black] (XX) circle (1.5pt);
	\draw[fill = black] (YY) circle (1.5pt);
	\draw[fill = black] (ZZ) circle (1.5pt);
	\node[below=.1 of MID,scale=1.1]{$G_0^1$};
	\end{scope}
	
	\begin{scope}[yshift=-70,xshift=360,local bounding box=graph9]
	\coordinate (II) at (-1,0);
	\coordinate (JJ) at (1,0);
	\coordinate (AA) at (0,1);
	\coordinate (XX) at (-.5,0);
	\coordinate (YY) at (.5,0);
	\coordinate (ZZ) at (0,.5);
	\coordinate (MID) at (0,-1);
	\draw[line width = 1pt,fill=lightgray] (0,0) circle (.5cm);
	\draw[ad](II) -- (XX);
	\draw[ad](YY) -- (JJ);
	\draw[ad](ZZ) -- (AA);
	\draw[fun](AA) arc (90:0:1);
	\draw[fun](JJ) arc (0:-180:1);
	\draw[fun](II) arc (180:90:1);
	\draw[fill = black] (XX) circle (1.5pt);
	\draw[fill = black] (YY) circle (1.5pt);
	\draw[fill = black] (ZZ) circle (1.5pt);
	\node[below=.1 of MID,scale=1.1]{$G_0^2$};
	\end{scope}
	
	\draw[-latex,thick]([yshift=-5]graph7.south) -- ([yshift=30]intersection of graph7--graph7.south and graph8--graph9);
	\path(graph8) -- (graph8-|graph9.west)  node[midway,below,scale=1.1]{\LARGE ,};
	
	\end{tikzpicture}
	\caption{Reductions of graphs with few external half-edges.}
	\label{figcolor8}
\end{figure}
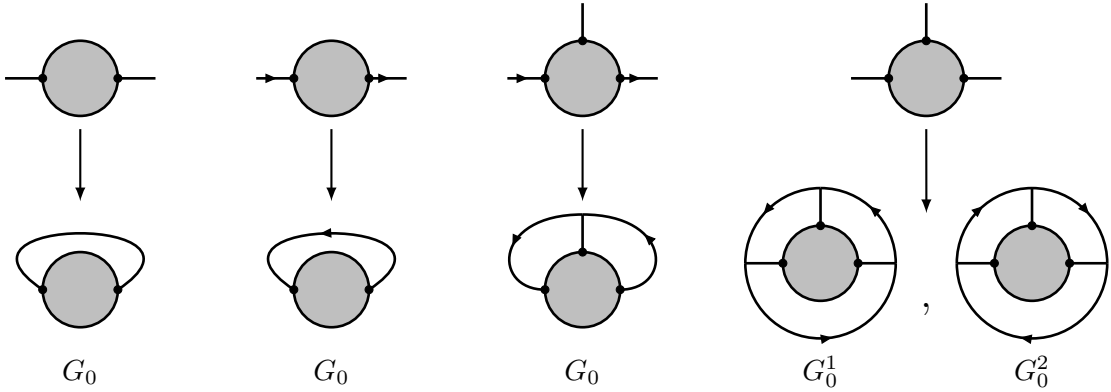

\begin{ex}\label{colorex3}
Consider the following examples for the reduction of a color graph $G$, see Figure \ref{figcolor8}.
\begin{enumerate}
\item If $G$ has two non-oriented external half-edges, the reduction has the form $R_G(N)_{i,j}=r_G(N) \delta_{i,j}$, where $ij$ is the non-oriented edge between the external half-edges.
Closing the edge gives the graph $G_0$ with reduction polynomial $R_{G_0}(N)$. From $\delta_{i,j}^2=N^2-1$ we get
\begin{equation}\label{eqcolorex1}
R_G(N)_{i,j}=\frac{R_{G_0}(N)}{N^2-1}\delta_{i,j}.
\end{equation}
\item If $G$ has two oriented external half-edges, the reduction has the form $R_G(N)_{a,b}=r_G(N)\delta_{a,b}$, where $ab$ is the oriented edge between the external half-edges.
Closing the edge gives the graph $G_0$ with reduction polynomial $R_{G_0}(N)$. From $\delta_{a,b}^2=N$ we get
\begin{equation}\label{eqcolorex2}
R_G(N)_{a,b}=\frac{R_{G_0}(N)}{N}\delta_{a,b}.
\end{equation}
\item If $G$ has two oriented external half-edges and one non-oriented external half-edge, the reduction has the form $R_G(N)^i_{a,b}=r_G(N)T^i_{ab}$.
We construct the vacuum graph $G_0$ by joining all external half edges in a fundamental vertex, i.e.\ we multiply by $T^i_{ba}$ and sum over $i$, $a$, and $b$.
Because $T^i_{ab}T^i_{ba}=N^2-1$ we obtain
\begin{equation}\label{eqcolorex3}
R_G(N)^i_{a,b}=\frac{R_{G_0}(N)}{N^2-1}T^i_{ab}.
\end{equation}
\item If $G$ has three non-oriented external half-edges, the reduction has two terms corresponding to two triangles with opposite orientation,
$$
R_G(N)^{i,j,k}=r_{G^1}(N)T^i_{ab}T^j_{bc}T^k_{ca}+r_{G^2}(N)T^k_{ab}T^j_{bc}T^i_{ca}.
$$
We close in two different ways by contraction with $T^i_{de}T^j_{ef}T^k_{fd}$ and with $T^k_{de}T^j_{ef}T^i_{fd}$, yielding
the vacuum graphs $G_0^1$ and $G_0^2$, respectively.
The transition matrix is given by two oriented triangles that are connected by three non-oriented edges. The orientation of the triangles can be parallel or opposite, yielding
the symmetric $2\times2$ matrix (see the last two graphs in Figure \ref{figcolor7})
$$
\Big(N-\frac1N\Big)\left(\begin{array}{cc}-2&N^2-2\\N^2-2&-2\end{array}\right).
$$
Inverting this matrix yields
\begin{equation}\label{eqcolorex4}
R_G(N)^{i,j,k}=\frac{(2R_{G_0^1}(N)+(N^2-2)R_{G_0^2}(N))T^i_{ab}T^j_{bc}T^k_{ca}+((N^2-2)R_{G_0^1}(N)+2R_{G_0^2}(N))T^k_{ab}T^j_{bc}T^i_{ca}}{N(N^2-1)(N^2-4)}.
\end{equation}
\end{enumerate}
\end{ex}
\begin{figure}[ht]
	\begin{tikzpicture}[scale=1.5]
	\coordinate (AA) at (-.8,-.375);
	\coordinate (BB) at (-.8,.375);
	\coordinate (CC) at (0,.8);
	\coordinate (DD) at (0,0);
	\coordinate (EE) at (0,-.8);
	\coordinate (XX) at (.666,.8);
	\coordinate (YY) at (.666,0);
	\coordinate (ZZ) at (.666,-.8);
	\draw[ad](AA) -- (CC);
	\draw[ad](AA) -- (DD);
	\draw[ad](AA) -- (EE);
	\draw[ad](BB) -- (CC);
	\draw[white,line width=4pt] (BB) -- (DD);
	\draw[ad](BB) -- (DD);
	\draw[white,line width=4pt] (BB) -- (EE);
	\draw[ad](BB) -- (EE);
	\draw[ad](CC) -- (XX);
	\draw[ad](DD) -- (YY);
	\draw[ad](EE) -- (ZZ);
	\draw[fill = black] (AA) circle (1.5pt);
	\draw[fill = black] (BB) circle (1.5pt);
	\draw[fill = black] (CC) circle (1.5pt);
	\draw[fill = black] (DD) circle (1.5pt);
	\draw[fill = black] (EE) circle (1.5pt);
	\end{tikzpicture}
	\caption{A subgraph which leads to zero color reduction.}
	\label{figcolor6}
\end{figure}
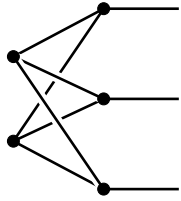

\subsection{Results and a conjecture}
\begin{prop}\label{propresults2}
Let $G$ with $v_G\geq1$ be a connected vacuum color graph for the group SU($N$) and let $R_G(N)$ be its color reduction.
Let $e_G$ and $f_G$ be the number of edges and oriented cycles in $G$, respectively.
\begin{enumerate}
\item $R_G(N)$ is divisible by $N^2-1$.
\item If $G$ has a subgraph as depicted in Figure \ref{figcolor6}, then $R_G(N)=0$.
\item The polynomial $R_G(N)\in\ZZ[N,\frac1N]$ has low degree $\geq -h_G+2$.
If $G$ has a fundamental vertex, then $\deg(R_G(N))\leq h_G$,  otherwise
$\deg(R_G(N))\leq h_G+1$.
\item The polynomial $R_G(N)$ has parity $f_G+e_G$,
\begin{equation}\label{eqparity}
R_G(-N)=(-1)^{f_G+e_G}R_G(N).
\end{equation}
\end{enumerate}
\end{prop}
\begin{proof}
(1) Because $v_G\geq1$, the graph $G$ has at least one non-oriented edge.
We open $G$ at this edge, see Example \ref{colorex3} (1), and calculate the reduction as a polynomial in $\ZZ[N,\frac1N]$. The result follows from (\ref{eqcolorex1}).

(2) By anti-symmetry of adjoint vertices, see Figure \ref{figcolor2}, permuting the two left vertices in Figure \ref{figcolor6} reproduces the graph with a minus sign
comming from the three right vertices.

(3) If $G$ has no adjoint vertex, then one third of its edges is non-oriented.
The reduction of a non-oriented edge with (\ref{color7}) produces one factor
of $1/N$ and may disconnect the graph. After all non-oriented edges are reduced,
we have $e_G/3$ powers of $1/N$ (bar cancellations). The number of oriented self-loops
(with value $N$) is between one and $e_G/3+1$.
From (\ref{halfedges}) we obtain $v_G=2e_G/3$; from (\ref{graphhom}) we get
$e_G/3=h_G-1$. The low degree of $R_G(N)$ is therefore $\geq-h_G+2$
and $\deg(R_G(N))\leq h_G$.

With the first equation in Figure \ref{figcolor3} we can replace an
adjoint vertex that is attached to an oriented chain by a fundamental vertex.
The reduction does not change $e_G$, $v_G$ or $c_G$. By (\ref{graphhom}) it
also fixes $h_G$. If $G$ has an adjoint and a fundamental vertex,
then (because $G$ is connected) there exists an adjoint vertex that is attached
to a fundamental vertex. By induction over the
number of adjoint vertices using the above reduction, we obtain a low degree $\geq-h_G+2$ and a degree $\leq h_G$.

If $G$ has no fundamental vertex, we use (\ref{color7}) to produce
three fundamental vertices. The loop order of $G$ increases by one.

(4) We can calculate the reduction only by using (\ref{color4}) (to eliminate all adjoint vertices) and (\ref{color7}) (to eliminate all non-oriented edges).

In both terms of (\ref{color4}), we obtain $f_G\mapsto f_G+1$, $e_G\mapsto e_G+3$. So, $f_G+e_G\mapsto f_G+e_G+4$ and the parity
does not change.

In the first term of (\ref{color7}) we obtain the map $f_G\mapsto f_G\pm1$
(depending on whether or not the oriented chains are from the same oriented cycle), $e_G\mapsto e_G-3$. So, $f_G+e_G$ does not change modulo 2.

In the second term of (\ref{color7}) we obtain the map $f_G\mapsto f_G$,
$e_G\mapsto e_G-3$. So, $f_G+e_G$ changes modulo 2.
The coefficient $1/N$ is anti-symmetric, so that (\ref{eqparity}) remains valid.

Hence, it suffices to show (\ref{eqparity}) for complete reductions,
which are a collection of oriented self-loops (with $e_\circ=0$).
A union of $n$ oriented loops has the reduction $N^n$ with
$n=f_G=f_G+e_G$.
\end{proof}
The graph in Figure \ref{figcolor6} is only the smallest example of a subgraph
that renders the color factor zero.

If $G_1$ ($G_2$) is a (non-)oriented cycle with $h_{G_1}-1$ ($h_{G_2}-1$) parallel chords, then
we obtain from repeatedly using the second line in Figure \ref{figcolor5} that
\begin{equation}
R_{G_1}(N)=\Big(N-\frac1N\Big)^{h_{G_1}-1}N,\quad R_{G_2}(N)=(2N)^{h_{G_2}-1}(N^2-1).
\end{equation}
This shows that the bounds for the degree and the low degree of $R_G(N)$ in Proposition \ref{propresults2} are sharp.

In the presence of adjoint vertices, many terms in the color reduction cancel. This leads to surprisingly simple results.
First examples of this phenomenon are the non-oriented cycles
with two and three edges in Figure \ref{figcolor5}. The coefficients of the
reduced graphs are $2N$ and $N$ (respectively) with no negative powers of $N$.
This seems to be a general feature of non-oriented color graphs.

\begin{con}
The color factor of a non-oriented vacuum graph $G$ with $v_G\geq4$ has low degree $\geq2$.
\end{con}

Because non-zero color factors for graphs with less than four vertices are powers of $N^2-1$, see Figure \ref{figcolor5}, the conjecture implies $R_G(N)\in\ZZ[N]$.
From Proposition \ref{propresults2} (4), it follows that the low degree of a connected non-oriented graph with even $h_G$ is odd, see (\ref{graphhom}) and (\ref{halfedges}).
In this case, the low degree of the color factor is conjectured to be $\geq3$.

Let $f_{\mathrm{loop}}(i_1,\ldots,i_n)$ be the non-oriented loop whose
vertices are attached to the external vertices $i_1,\ldots i_n$ in counter-clockwise order.
Likewise, let $T_{\mathrm{loop}}(i_1,\ldots,i_n)$ be the oriented loop
with external vertices $i_1,\ldots i_n$.
From (\ref{color8}) we, e.g., get
\begin{equation}\label{smallloops}
T_{\mathrm{loop}}(i_1,i_2)=\delta_{i_1,i_2},\quad f_{\mathrm{loop}}(i_1,i_2)=2N\delta_{i_1,i_2},\quad f_{\mathrm{loop}}(i_1,i_2,i_3)=Nf_{i_1,i_2,i_3}.
\end{equation}
For more than two indices, $T_{\mathrm{loop}}$ cannot be reduced.
Equation (\ref{color4}) (bottom identity in Figure \ref{figcolor3}) is
\begin{equation}\label{fTred}
f_{i_1,i_2,i_3}=T_{\mathrm{loop}}(i_1,i_2,i_3)-T_{\mathrm{loop}}(i_3,i_2,i_1).
\end{equation}
The following proposition gives an antipode-like reduction formula for $f_{\mathrm{loop}}$.

\begin{prop}\label{prop:results}
For $n\geq2$ we have (to lighten the notation, we use numbers for external labels)
\begin{equation}\label{floop}
f_{\mathrm{loop}}(1,\ldots,n)=\sum_{1,\ldots,n=S\sqcup T}(-1)^{|T|}T_{\mathrm{loop}}(S)\,T_{\mathrm{loop}}(\tilde T),
\end{equation}
where the sum is over all $2^n$ ordered (we distinguish between $S\sqcup T$ and $T\sqcup S$) partitions of $1,\ldots,n$ into $S$ and $T$ which
are in natural order. Moreover, $|T|$ is the cardinality of $T$ and $\tilde T$ is $T$ in reverse order.
\end{prop}
\begin{proof}
Let $v_i$ be the vertex in $f_{\mathrm{loop}}(1,\ldots,n)$ that is attached to $i$, $i=1,\ldots,n$. Using (\ref{color4}) at the vertex $v_n$
gives a non-oriented loop with an oriented insertion. We write the result as
$$
f_{\mathrm{loop}}(1,\ldots,n)=fT_{\mathrm{loop}}(1,\ldots,n-1;v_1,v_{n-1},n)-fT_{\mathrm{loop}}(1,\ldots,n-1;n,v_{n-1},v_1).
$$
Both lists, $1,\ldots,n-1$ for the non-oriented loop and $v_1,v_{n-1},n$ or $n,v_{n-1},v_1$ for the oriented insertion, are in counter-clockwise order.

The reduction of the vertex $v_2$ with the first identity in Figure \ref{figcolor3}
gives (note that, due to the counter-clockwise orientation of the insertion, the oriented line in Figure \ref{figcolor3} has to be reversed, so that the cross term has negative sign)
$$
fT_{\mathrm{loop}}(1,2,\ldots,n-1;v_1,v_{n-1},n)=fT_{\mathrm{loop}}(2,\ldots,n-1;1,v_2,v_{n-1},n)-fT_{\mathrm{loop}}(2,\ldots,n-1;v_2,1,v_{n-1},n).
$$
By iteration we get
$$
fT_{\mathrm{loop}}(1,2,\ldots,n-1;v_1,v_{n-1},n)=\sum_{1,\ldots,k-1=S\sqcup T}(-1)^{|T|}fT_{\mathrm{loop}}(k,\ldots,n-1;S,v_k,\tilde T,v_{n-1},n).
$$
Reduction of the last adjoint vertex $v_{n-1}$ gives
$$
fT_{\mathrm{loop}}(n-1;S,v_{n-1},\tilde T,v_{n-1},n)=T_{\mathrm{loop}}(S,n-1,v_{n-1},\tilde T,v_{n-1},n)-T_{\mathrm{loop}}(S,v_{n-1},n-1,\tilde T,v_{n-1},n),
$$
where the $v_{n-1}$ stand for pairs of vertices in the oriented loop that are connected by a non-oriented edge. We obtain
$$
fT_{\mathrm{loop}}(1,2,\ldots,n-1;v_1,v_{n-1},n)=\sum_{1,\ldots,n-1=S\sqcup T}(-1)^{|T|}T_{\mathrm{loop}}(S,v_{n-1},\tilde T,v_{n-1},n).
$$
Reduction of the non-oriented edge $v_{n-1}v_{n-1}$ using (\ref{color7}) (see Figure \ref{figcolor4}) yields
$$
T_{\mathrm{loop}}(S,v_{n-1},\tilde T,v_{n-1},n)=
T_{\mathrm{loop}}(n,S)\,T_{\mathrm{loop}}(\tilde T)-\frac1N\,T_{\mathrm{loop}}(S,\tilde T,n).
$$
In the second term, the decompositions $1,\ldots,n-1=S_0,n-1\sqcup T_0=S_0\sqcup T_0,n-1$ cancel in the signed sum over $S\sqcup T$. We hence obtain
$$
f_{\mathrm{loop}}(1,\ldots,n)=\sum_{1,\ldots,n-1=S\sqcup T}(-1)^{|T|}\big(T_{\mathrm{loop}}(n,S)\,T_{\mathrm{loop}}(\tilde T)-T_{\mathrm{loop}}(\tilde S,n)\,T_{\mathrm{loop}}(T)\big),
$$
where the second term comes from reversing the order of the orientation in the insertion. Upon swapping $S$ and $T$ in the second term of the sum,
both terms combine to yield the desired result.
\end{proof}
Notice that $T_{\mathrm{loop}}(i)=0$, so that partitions with $|S|=1$ or $|T|=1$ can be omitted from the sum. For the empty set we have $T_{\mathrm{loop}}(\emptyset)=N$.

\begin{ex}
The cases $n=2$ and $n=3$ are in (\ref{smallloops}) and (\ref{fTred}).
For $n=4$ we have the nontrivial decompositions $1,2,3,4\sqcup\emptyset$; $\emptyset\sqcup1,2,3,4$; $1,2\sqcup3,4$; $1,3\sqcup2,4$; $1,4\sqcup2,3$; $2,3\sqcup1,4$; $2,4\sqcup1,3$; $3,4\sqcup1,2$. With the third line in
Figure \ref{figcolor3} we obtain 
\begin{equation}
f_{\mathrm{loop}}(i_1,i_2,i_3,i_4)=NT_{\mathrm{loop}}(i_1,i_2,i_3,i_4)+NT_{\mathrm{loop}}(i_4,i_3,i_2,i_1)
+2\delta_{i_1,i_2}\delta_{i_3,i_4}+2\delta_{i_1,i_3}\delta_{i_2,i_4}
+2\delta_{i_1,i_4}\delta_{i_2,i_3}.
\end{equation}
\end{ex}

\section{Gamma reduction}
The reduction of traces of $\gamma$ matrices is standard, see e.g.\ \cite{IZ,Form}.
We have the anti-commutator relation
\begin{equation}\label{gammadelta}
\{\gamma_\alpha,\gamma_\beta\}=2\delta_{\alpha,\beta}1,
\end{equation}
where $1$ is the unit matrix in the vector space of the $\gamma$ matrices.
Moreover, the $\gamma$ matrices are traceless,
\begin{equation}\label{gamma0}
\Tr \gamma_\alpha=0.
\end{equation}
The dimension of space(-time) is
\begin{equation}
D=\delta_{\alpha,\alpha}.
\end{equation}
In QFT, the dimension $D$ is sometimes considered as non-integer parameter.
The following results are consistent with this setup.

We define chains of gamma matrices
$$
S_n=\gamma_{\alpha_1}\gamma_{\alpha_2}\cdots\gamma_{\alpha_n}.
$$
Upper indices indicate $\gamma$ matrices that are omitted in $S_n$,
\begin{align*}
S_n^k&=\gamma_{\alpha_1}\cdots\gamma_{\alpha_{k-1}}\gamma_{\alpha_{k+1}}\cdots\gamma_{\alpha_n}\equiv\gamma_{\alpha_1}\cdots\widehat{\gamma_k}\cdots\gamma_{\alpha_n},\\
S_n^{k,\ell}&=\gamma_{\alpha_1}\cdots\widehat{\gamma_k}\cdots\widehat{\gamma_\ell}\cdots\gamma_{\alpha_n}.
\end{align*}
With this notation the anti-commutator iterates to
\begin{equation}\label{gamma1}
\gamma_\beta S_n=2\sum_{k=1}^n(-1)^{k-1}\delta_{\alpha_k,\beta}S_n^k+(-1)^nS_n\gamma_\beta.
\end{equation}
Left and right multiplication with $\gamma_\beta$ gives (respectively)
\begin{equation}\label{gamma2}
\gamma_\beta S_n\gamma_\beta=(-1)^n(D-2)S_n+2\sum_{k=2}^n(-1)^{n-k}\gamma_{\alpha_k}S_n^k=(-1)^n(D-2)S_n+2\sum_{k=1}^{n-1}(-1)^{k-1}S_n^k\gamma_{\alpha_k}.
\end{equation}
By anti-commuting $\gamma_{\alpha_1}$ and $\gamma_{\alpha_2}$ in the terms $k=2$ and $k=3$ of
the first identity in (\ref{gamma2}) we obtain for $n\geq3$ the one term shorter relation
\begin{equation}\label{gamma2b}
\gamma_\beta S_n\gamma_\beta=(-1)^n\Big((D-4)S_n+2\gamma_{\alpha_3}\gamma_{\alpha_2}S_n^{2,3}+2\sum_{k=4}^n(-1)^k\gamma_{\alpha_k}S_n^k\Big).
\end{equation}
Likewise, we get from the second identity in (\ref{gamma2})
\begin{equation}\label{gamma2c}
\gamma_\beta S_n\gamma_\beta=(-1)^n\Big((D-4)S_n+2S_n^{n-2,n-1}\gamma_{\alpha_{n-1}}\gamma_{\alpha_{n-2}}\Big)+2\sum_{k=1}^{n-3}(-1)^{k-1}S_n^k\gamma_{\alpha_k}.
\end{equation}
Let $\tilde S_n=\gamma_{\alpha_n}\cdots\gamma_{\alpha_1}$ be $S_n$ in reversed order. We obtain the following contraction formulae
\begin{align}\label{gamma3}
\gamma_\beta^2&=D\,1,\\\nonumber
\gamma_\beta S_1\gamma_\beta&=-(D-2)S_1,\\\nonumber
\gamma_\beta S_2\gamma_\beta&=(D-4)S_2+4\delta_{\alpha_1,\alpha_2},\\\nonumber
\gamma_\beta S_3\gamma_\beta&=-(D-4)S_3-2\tilde S_3,\\\nonumber
&=-(D-6)S_3-4\delta_{\alpha_1,\alpha_2}\gamma_3+4\delta_{\alpha_1,\alpha_3}\gamma_2-4\delta_{\alpha_2,\alpha_3}\gamma_1,\\\nonumber
\gamma_\beta S_4\gamma_\beta&=(D-4)S_4+2\gamma_{\alpha_3}\gamma_{\alpha_2}\gamma_{\alpha_1}\gamma_{\alpha_4}+2\gamma_{\alpha_4}\gamma_{\alpha_1}\gamma_{\alpha_2}\gamma_{\alpha_3}\\\nonumber
&=(D-6)S_4-2\tilde S_4+8(\delta_{\alpha_1,\alpha_2}\delta_{\alpha_3,\alpha_4}-\delta_{\alpha_1,\alpha_3}\delta_{\alpha_2,\alpha_4}+\delta_{\alpha_1,\alpha_4}\delta_{\alpha_2,\alpha_3})\\\nonumber
&=(D-8)S_4+4(\delta_{\alpha_1,\alpha_2}\gamma_{\alpha_3}\gamma_{\alpha_4}-\delta_{\alpha_1,\alpha_3}\gamma_{\alpha_2}\gamma_{\alpha_4}+\delta_{\alpha_1,\alpha_4}\gamma_{\alpha_2}\gamma_{\alpha_3}
+\delta_{\alpha_2,\alpha_3}\gamma_{\alpha_1}\gamma_{\alpha_4}-\delta_{\alpha_2,\alpha_4}\gamma_{\alpha_1}\gamma_{\alpha_3}\\\nonumber
&\qquad\qquad\qquad\qquad+\delta_{\alpha_3,\alpha_4}\gamma_{\alpha_1}\gamma_{\alpha_2}).
\end{align}

In the following lemma we summarize more properties of $\gamma$ products and $\gamma$ traces.
\begin{lem}
We sum over iterated indices and assume that the dimension $D$ is not an odd integer. Then
\begin{enumerate}
\item
\begin{equation}\label{gamma2a}
\gamma_\beta S_n\gamma_\beta=(-1)^n(D-2n)S_n-4\sum_{1\leq k<\ell\leq n}(-1)^{n+k+\ell}\delta_{\alpha_k,\alpha_\ell}S_n^{k,\ell}.
\end{equation}
\item
\begin{equation}\label{gamma4}
\Tr S_n=\begin{cases}
\sum_{k=2}^n(-1)^k\delta_{\alpha_1,\alpha_k}\Tr S_n^{1,k}&\text{if $n$ even},\\
0&\text{if $n$ odd.}
\end{cases}
\end{equation}
\item
\begin{equation}\label{gamma5}
\Tr S_n=\Tr\tilde S_n.
\end{equation}
\end{enumerate}
\end{lem}
\begin{proof}
(1) We substitute (\ref{gamma1}) for $\gamma_{\alpha_k}S_{k-1}=\gamma_{\alpha_k}\gamma_{\alpha_1}\cdots\gamma_{\alpha_{k-1}}$ into the first equation of (\ref{gamma2}) yielding
$$
\gamma_\beta S_n\gamma_\beta=(-1)^n(D-2)S_n+2\sum_{k=2}^n(-1)^{n-k}(-1)^{k-1}S_n+4\sum_{1\leq k<\ell\leq1}(-1)^{n-\ell+k-1}\delta_{\alpha_k,\alpha_\ell}S_n^{k,\ell}.
$$
The first sum is $-2(n-1)(-1)^nS_n$ and the result follows.

(2) For even $n$ we use (\ref{gamma1}) for $S_n=\gamma_{\alpha_1}S_n^1$. The result follows from the cyclicity of the trace.

For odd $n$ the proof is by induction over $n$. The case $n=1$ is (\ref{gamma0}).
For $n\geq3$ we use (\ref{gamma2a}) which simplifies by induction and by cyclicity of the trace to $D\Tr S_n=-(D-2n)\Tr S_n$.
Because $D$ is not odd, we have $D\neq n$ and the result follows.

(3) If $n$ is odd, then (\ref{gamma5}) is trivial. For even $n$ we use induction over $n$. For $n=2$ the result follows from the cyclicity of the trace.
For $n\geq4$ we use (\ref{gamma4}) for $\gamma_{\alpha_1}\gamma_{\alpha_n}\cdots\gamma_{\alpha_2}$. We obtain by induction
$$
\Tr\tilde S_n=\sum_{k=2}^n(-1)^k\delta_{\alpha_1,\alpha_{n+2-k}}\Tr S_n^{1,n+2-k}.
$$
After $k\mapsto n+2-k$ the result follows from (\ref{gamma4}).
\end{proof}

The algorithm for calculating $\Tr S_n$ is evident. If $n$ is odd, then $\Tr S_n=0$. Otherwise, we search for the smallest $r$ (if existent) such that $S_n$
or any of its cyclic permutations has a sequence $\gamma_{\beta}\gamma_{\beta_1}\cdots\gamma_{\beta_r}\gamma_{\beta}$ for distinct $\beta,\beta_1,\ldots,\beta_r$
in $\{\alpha_1,\ldots,\alpha_n\}$. We simultaneously reduce all cases with $r=0,1$, see (\ref{gamma3}). For $r\geq2$, we use (\ref{gamma2a}) for an iterative reduction.
If $r$ does not exist because all $\gamma$ matrices in the trace have distinct indices, then we use (\ref{gamma4}).
In this case, the formula for a complete reduction is universal for a fixed number $n$ of $\gamma$ matrices. This allows us to cache all results up to $n=16$.
Moreover, we use (\ref{gammadelta}) and (\ref{gamma2a}) to order $\gamma$ products without trace.

One can further improve the algorithm by caching all $\gamma$ traces that are not amenable to the first reduction step with $r=0,1$.
In this approach, reduction with (\ref{gamma2b}) insead of (\ref{gamma2a}) can be preferable because the former generates less terms.
This has not (yet) been implemented.

In a QFT with $\gamma$ matrices in vertices, we benefit from the fact that in a single ferminon loop, half of the gamma matrices are contracted.
At loop order $\ell$ we are left with $2\ell$ indices that are not contracted.
In the case of several fermion loops, it is important to start with the evaluation of
the trace with the smallest number of uncontracted indices.
There exists a loop with $\leq2\ell$ uncontracted indices, so that the degree in Kronecker $\delta$'s is always $\leq\ell$.

We tested the Maple implementation {\tt HyperlogProcedures} \cite{Shlog} by
calculating the traces in Feynman graphs that contribute to the photon propagator.
The average computation time on a single core of an office PC
is approximately two minutes for $\ell=6$ loops and 30 minutes for $\ell=7$. 
For graphs with $\ell\leq 5$, the calculation is nearly instant.

\end{document}